\documentclass[acmsmall,nonacm]{acmart}

\usepackage{algorithm}
\usepackage[noend]{algpseudocode}
\usepackage{wrapfig}
\usepackage{thm-restate}

\newtheorem{theorem}{Theorem}
\newtheorem{definition}[theorem]{Definition}
\newtheorem{corollary}[theorem]{Corollary}

\newtheorem{lemma}[theorem]{Lemma}

\title{Strong Linearizability using Primitives with Consensus Number 2}

\author{Hagit Attiya}
\orcid{0000-0002-8017-6457}
\affiliation{
\institution{Technion}
\department{Department of Computer Science}
\country{Israel}
}
\email{hagit@cs.technion.ac.il}

\author{Armando Castañeda}
\orcid{0000-0002-8017-8639}
\affiliation{
\institution{Univesidad Nacional Aut\'onoma de M\'exico}
\department{Instituto de Matem\'aticas}
\country{Mexico}
}
\email{armando.castaneda@im.unam.mx}

\author{Constantin Enea}
\orcid{0000-0003-2727-8865}
\affiliation{%
\institution{LIX, Ecole Polytechnique, CNRS and Institut Polytechnique de Paris}
\country{France}
}
\email{cenea@lix.polytechnique.fr}

\begin{document}

\begin{abstract}
A powerful tool for designing complex concurrent programs is through
composition with \emph{object implementations} from lower-level primitives.
\emph{Strongly-linearizable} implementations allow to
preserve \emph{hyper-properties}, e.g.,
probabilistic guarantees of randomized programs.
However, the only known wait-free strongly-linearizable implementations
for many objects rely on {\sf compare\&swap},
a universal primitive that allows any number of processes to solve consensus.
This is despite the fact that these objects
have wait-free linearizable implementations from
read / write primitives, which do not support consensus.
This paper investigates a middle-ground,
asking whether there are wait-free strongly-linearizable implementations
from realistic primitives such as {\sf test\&set} or {\sf fetch\&add},
whose consensus number is 2.

We show that many objects with consensus number 1
have wait-free strongly-linearizable implementations
from {\sf fetch\&add}.
We also show that several objects with consensus number 2
have wait-free or lock-free implementations
from other objects with consensus number 2.
In contrast, we prove that even when {\sf fetch\&add}, swap and {\sf test\&set}
primitives are used, some objects with consensus number 2
do not have lock-free strongly-linearizable implementations.
This includes queues and stacks, as well as relaxed variants thereof.
\end{abstract}

\maketitle

\section{Introduction}

A key way to construct complex distributed systems is through modular
composition of linearizable concurrent objects~\cite{HerlihyW1990}.
Yet linearizable objects do not always compose correctly
with randomized programs~\cite{HadzilacosHT2020,GolabHW2011},
or with programs that should not leak information~\cite{AttiyaE2019}.
This deficiency is addressed by \emph{strong linearizability}~\cite{GolabHW2011},
a restriction of linearizability,
which ensures that properties holding when a concurrent program
is executed in conjunction with an atomic object,
continue to hold when the program is executed
with a strongly-linearizable implementation of the object.

More generally,
strong linearizability was shown~\cite{AttiyaE2019,DongolSW2023}
to preserve \emph{hyperproperties}~\cite{ClarksonS2010},
such as security properties and
probability distributions of reaching particular program states.
This made strongly-linearizable concurrent objects very sought after.

Unfortunately, the only known \emph{wait-free} strongly-linearizable
implementations, in which every operation completes, use
primitives such as {\sf compare\&swap}~\cite{GolabHW2011,HwangW2021}.
These primitives have infinite \emph{consensus number}, in the sense
that they allow to solve consensus for any number of processes.\footnote{
    See a formal definition of the \emph{consensus number}
    in Section~\ref{sec:preliminaries}.}
This makes them \emph{universal} as they can be used
to implement virtually any shared object~\cite{H91}.

Weaker primitives, with consensus number 1, e.g., read and write,
do not admit strongly-linearizable implementations:
Many frequently-used concurrent objects that have wait-free linearizable
implementations from common read and write primitives, are known
not to have analogous strongly-linearizable implementations.
For example, max registers, snapshots, or monotonic counters
do not have wait-free strongly-linearizable implementations,
even with multi-writer registers~\cite{DenysyukW2015}.
Single-writer registers do not suffice even for \emph{lock-free}
strongly-linearizable multi-writer registers, max registers,
snapshots, or counters~\cite{HelmiHW2012}.

In between universal primitives, like {\sf compare\&swap},
and weak primitives with consensus number 1,
like {\sf read} and {\sf write},
there are primitives, like {\sf test\&set}, {\sf fetch\&add}
and {\sf swap}, whose consensus number is 2.
These are realistic primitives, provided in many architectures.
This paper seeks a middle ground, investigating whether primitives
with consensus number 2 allow to obtain wait-free,
or at least lock-free,
strongly-linearizable implementations.

Our first set of results show that many objects with consensus number 1
have wait-free strongly-linearizable implementations,
when {\sf fetch\&add} can be used (see Section~\ref{sec:algs cons 1}).
Our construction goes through a novel, but simple, implementation
of wait-free strongly-linearizable atomic snapshots.

\begin{wrapfigure}{r}{0.5\textwidth}
\centering
\includegraphics[scale=0.7]{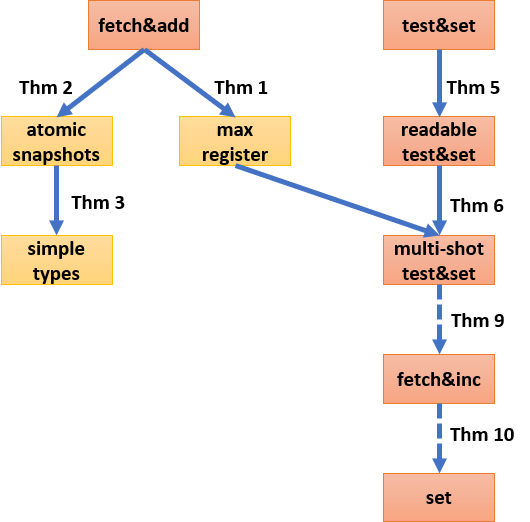}
\caption{Summary of our constructions:
objects with consensus number 2 appear in orange while
objects with consensus number 1 appear in yellow;
solid arrows indicate wait-free implementations,
while dashed arrows indicate lock-free implementations.}
\label{figure:implementations}
\end{wrapfigure}

With strongly-linearizable atomic snapshots at our disposal,
we follow the scheme of Aspnes and Herlihy~\cite{AspnesH1990wait}
and use atomic snapshot to implement \emph{simple types},
where all operations either commute (their order is immaterial)
or one of them overwrites the other (it is immaterial whether
the first operation is executed immediately before the second operation,
or not).
For example, max registers~\cite{AspnesAC2009},
in which a {\sf ReadMax} returns the largest value previously written by a {\sf WriteMax}.
{\sf ReadMax} operations commute,
{\sf WriteMax} operations overwrite {\sf ReadMax} operations,
and {\sf WriteMax}($v_1$) overwrites {\sf WriteMax}($v_2$),
if $v_1 \geq v_2$.
Counters, logical clocks and certain set objects are also simple types.
Ovens and Woelfel~\cite{OvensW19} have shown that using
strongly-linearizable snapshots make this implementation
also strongly linearizable.
We provide a simple and elegant proof of this fact,
relying on a forward simulation argument,
which is sufficient for strong linearizability~\cite{AttiyaE2019}.

We next consider objects with consensus number 2,
and show that some of them have wait-free or lock-free implementations
from other objects with consensus number 2 (see Section~\ref{sec:algs cons 2}).
We show that {\sf test\&set} can wait-free strongly-linearizable implement
readable {\sf test\&set}, i.e., the usual {\sf test\&set} enriched with a read operation.
Moreover, it can lock-free implement multi-shot {\sf test\&set}, which provides a reset operation,
and if in addition one considers max registers, wait-freedom is possible.
Also, we show lock-free strongly-linearizable implementations of fecth\&increment and sets
using only {\sf test\&set}. These constructions appear in Section~\ref{sec:algs cons 2}.
Figure~\ref{figure:implementations} summarizes our positive results.

In contrast, we prove that even when {\sf fetch\&add}, {\sf swap} and {\sf test\&set}
primitives are used, some objects with consensus number 2
do not have lock-free strongly-linearizable
implementations.
This includes objects like queues and stacks,
as well as several relaxed variants thereof
(see Section~\ref{sec:imposs cons 2}).
The proof goes through showing a connection between lock-free strongly-linearizable
implementations of these objects and $k$-set agreement: a single instance
of any such implementations allows $n$-processes to solve $k$-set agreement.
For $n \geq 3$ and $k=1$,
consensus is known to be impossible using objects with consensus number 2~\cite{H91},
which prevents strongly-linearizable implementations with such base objects.
Similarly, whenever $n > 2k$, there is no $k$-set agreement algorithm for $n$ processes using
$2$-process {\sf test\&set}~\cite{HerlihyR94},
and hence in these cases there is no strongly-linearizable implementation using only this base object.
The connection between lock-free strong linearizable implementations and $k$-set agreement
we present here, is motivated by the simulation in~\cite{AttiyaCH2018} that solves
consensus from any lock-free linearizable implementation of a queue with \emph{universal helping},
which, roughly speaking, means that eventually every operation, complete or pending, is linearized.

\subsubsection*{Additional Related Work.}
\label{section:related}

Golab, Higham and Woelfel~\cite{GolabHW2011} were the first
to recognize the problem when linearizable
objects are used with randomized programs, via an
example using the snapshot object implementation of~\cite{AADGMS93}.
They proposed \emph{strong linearizability}
as a way to overcome the increased vulnerability
of programs using linearizable implementations to strong adversaries,
by requiring that the linearization order of operations
at any point in time be consistent with the linearization order
of each prefix of the execution.
Thus, strongly-linearizable implementations limit the adversary's
ability to gain additional power by manipulating
the order of internal steps of different processes.
Consequently, properties holding when a
concurrent program is executed with an atomic object,
continue to hold when the program is executed with a strongly-linearizable
implementation of the object (see~\cite{AttiyaE2019,DongolSW2023}).

Most prior work on strong linearizability focused on implementations
using shared objects, and considered various progress properties.
The exception are~\cite{AttiyaEW2021,ChanHHT2021},
who studied message-passing implementations.

If one only requires \emph{obstruction-freedom}, which ensures an operation
complete only if it executes alone, then any object can be implemented
using single-writer registers~\cite{HelmiHW2012}.

When considering the stronger property of \emph{lock-freedom},
which requires that as long as some operation is pending, some
operation completes, single-writer registers are not sufficient
for implementing multi-writer registers, max registers, snapshots,
or counters~\cite{HelmiHW2012}.
If the implementations can use multi-writer registers, though, it
is possible to get lock-free implementations of max registers,
snapshots, and monotonic counters~\cite{DenysyukW2015}, as well as of
objects whose operations commute or overwrite~\cite{OvensW19}.
It was also shown~\cite{AttiyaCH2018} that there is no
lock-free implementation of a queue or a stack with universal helping,
from objects whose readable versions have consensus number less than the
number of processes, e.g., readable {\sf test\&set}.

For the even stronger property of {\em wait-freedom}, which requires every
operation to complete, it is possible to implement bounded max registers
using multi-writer registers~\cite{HelmiHW2012}, but it is impossible
to implement max registers, snapshots, or monotonic
counters~\cite{DenysyukW2015} even with multi-writer registers.
The bottom line is that the only known strongly-linearizable wait-free
implementation is of a bounded max register (using multi-writer registers),
while many impossibility results are known.

There are lock-free or wait-free implementations
of objects with consensus number 2 from objects at the same level of the consensus
hierarchy~\cite{AfekWW1993, Afek1999W, AGM07, AfekMW2011, CastanedaRR23, Li01}.
Attiya and Enea~\cite{AttiyaE2019} have already shown, by example,
that the wait-free stack implementation~\cite{AGM07} is not strongly linearizable,
but it was not discussed whether the other implementations
are strongly linearizable or not.
Our impossibility result implies that the lock-free queue and wait-free stack implementations in~\cite{AGM07, Li01},
based on {\sf test\&set}, {\sf swap} and {\sf fetch\&add}, are not strongly linearizable,
and neither are the read/write lock-free and wait-free (relaxed)
queue and stack implementations with multiplicity in~\cite{CastanedaRR23}.

The wait-free one-shot {\sf fetch\&increment} using {\sf test\&set}
in~\cite{AfekWW1993, Afek1999W} is strongly linearizable;
our lock-free {\sf fetch\&increment} strongly-linearizable implementation
is a straightforward generalization of that implementation.
In contrast, the wait-free multi-shot {\sf fetch\&increment} implementation
in the same work is not strongly linearizable as
it has executions in which an operation returns $x$,
and pending operations certainly will return a value $< x$ but it is
not decided the actual value yet, and the return value depends on the extension.
A similar situation happens in some executions of the wait-free
multi-shot {\sf fetch\&add} implementation in~\cite{AfekWW1993, Afek1999W}
and the wait-free {\sf swap} implementation in~\cite{AfekMW2011}.

\section{Preliminaries}
\label{sec:preliminaries}

We consider a standard shared memory system with $n$ asynchronous processes, $p_0,\ldots,p_{n-1}$,
which may crash at any time during an execution.
Processes communicate with each other by applying \emph{atomic}
operations to shared \emph{base objects}.

A \emph{(high-level) concurrent object}
is defined by a state machine consisting of a set of states,
a set of operations, and a set of transitions between states.
Such a specification is known as \emph{sequential}.
An \emph{implementation} of an object $T$ is a distributed
algorithm $\mathcal{A}$ consisting of a local state machine
$\mathcal{A}_p$, for each process $p$.
$\mathcal{A}_p$ specifies which operations of base objects $p$ applies
and which local computations $p$ performs in order to return a response
when it invokes an operation of $T$.
Each of these base object operation invocations and local computations is a \emph{step}.
For the rest of this section,
fix an implementation $\mathcal{A}$ of an object $T$.

A \emph{configuration} $C$ of the system contains the states of
all shared base objects and processes.
In an \emph{initial} configuration,
base objects and processes are in their initial states.
Given a configuration $C$ and a process $p$, $p(C)$ is the configuration
after $p$ takes its next step in $C$.
Moreover, $p^{0}(C) = C$ and for every $n\in \mathbb{N}$,
$p^{n+1}(C) = p(p^{n}(C))$.
Note that the next step $p$ takes in $C$ depends only on
its local state in $C$.

An \emph{execution of $\mathcal{A}$ starting from $C_0$}
is a (possibly infinite) sequence
$C_0 e_1 C_1 e_2 C_2 \cdots$, where each $e_i$ is a step of a process,
or an invocation/response of a high-level operation by a process
and if $e_i$ is a step, then $C_{i+1} = e_i(C_i)$;
furthermore, the sequence satisfies the following properties:
\begin{enumerate}
\item Each process can invoke a new (high-level) operation only
 when its previous operation (if there is any) has a corresponding response,
 i.e., executions are \emph{well-formed}.
\item A process takes steps only between an invocation and a response.
\item For any invocation of process $p$,
  the steps of $p$ between that invocation and the following response of $p$,
  if there is one,
  correspond to steps of $p$ that are specified by $\mathcal{A}_p$.
\end{enumerate}

An execution $\beta$ is an \emph{extension} of a finite execution $\alpha$
if $\alpha$ is a prefix of $\beta$.
A configuration $C$ is \emph{reachable} if there is a finite execution
$\alpha$ starting from an initial configuration whose last configuration is $C$;
we say that $\alpha$ \emph{ends} with $C$.
A configuration $C'$ is \emph{reachable} from a configuration $C$
if there is a finite execution starting from $C$ that ends with $C'$.
Two configurations $C$ and $C'$ are \emph{indistinguishable} to process $p$
if the state of every base object and the state of $p$
are the same in $C$ and $C'$.

An operation in an execution is \emph{complete} if both its invocation
and response appear in the execution.
An operation is \emph{pending} if only its invocation appears in the execution.
An implementation is \emph{wait-free} if every process completes each of
its operations in a finite number of its steps.
Formally, if a process executes infinitely many steps in an execution,
it completes infinitely many operations.
An implementation is \emph{lock-free} if whenever processes
execute steps, at least one of the operations terminates.
Formally, in every infinite execution,
infinitely many operations are complete.
Thus, a wait-free implementation is lock-free
but not necessarily vice versa.

In an execution, an operation $OP$ \emph{precedes}
another operation $OP'$ if
the response of $OP$ appears before the invocation of $OP'$.
$OP$ and $OP'$ are \emph{overlapping} if neither $OP$ precedes $OP'$
nor $OP'$ precedes $OP$.
$OP$ \emph{does not precede} $OP'$ if
either $OP$ and $OP'$ are overlapping or $OP'$ precedes $OP$.

\emph{Linearizability}~\cite{HerlihyW1990} is the standard notion used
to identify a correct implementation.
Roughly speaking, an implementation is linearizable if each operation
appears to take effect \emph{atomically} at some time between its invocation
and response, hence operations' real-time order is maintained.
Formally, let $\mathcal{A}$ be an implementation of an object $T$.
An execution $\alpha$ of $\mathcal{A}$ is \emph{linearizable} if
there is a sequential execution $S$ of $T$
(i.e., a sequence of matching invocation-response pairs, starting with an invocation)
such that:
\begin{enumerate}
\item
$S$ contains every complete operation in $\alpha$ and
some of the pending operations in $\alpha$.
Hence, the output values in the matching responses of an invocation in
$S$ and the complete operations in $\alpha$ are the same.
\item
If $OP$ precedes $OP'$ in $\alpha$,
then $OP$ precedes $OP'$ in $S$;
namely, $S$ respects the \emph{real-time} order in $\alpha$.
\end{enumerate}
$\mathcal{A}$ is \emph{linearizable} if all its executions are linearizable.

Roughly speaking, an implementation of a data type is \emph{strongly
linearizable}~\cite{GolabHW2011} if once an operation is linearized,
its linearization order cannot be changed in the future.
More specifically,
there is a function $L$ mapping each execution to a linearization,
and the function is \emph{prefix-closed}:
for every two executions $\alpha$ and $\beta$,
if $\alpha$ is a prefix of $\beta$, then $L(\alpha)$ is a prefix of $L(\beta)$.

A \emph{$k$-set agreement} object, $1 \leq k \leq n$, provides a single operation,
called ${\sf decide(\cdot)}$;
each process can invoke  ${\sf decide(\cdot)}$ once,
with its \emph{proposal} to the consensus as input.
The processes obtain a value from the object, called \emph{the decision},
so the following properties are satisfied in every execution:
\begin{description}

\item[Termination.] Every correct process decides after a finite number of steps.

\item[Validity.] Processes decide only on proposed values
    (i.e., an input to ${\sf decide(\cdot)}$).

\item[$k$-Agreement.] Processes decide on at most $k$ distinct values.

\end{description}

The well-known consensus problems is $1$-set agreement.
Consensus is \emph{universal}~\cite{H91} in the sense that
objects solving consensus among $n$ processes,
together with read/write registers,
provide a wait-free $n$-process linearizable implementation
of any concurrent object with a sequential specification.
The \emph{consensus number} of a shared object~\cite{H91}
is the maximum number $n$
such that it is possible to solve consensus among $n$ processes
from read/write registers and instances of the object.

\section{Strongly-linearizable implementations of objects with consensus number 1}
\label{sec:algs cons 1}

We describe a recipe for strongly-linearizable (and wait-free)
implementation of objects with consensus number 1, using {\sf fetch\&add},
an object with consensus number 2.
As a warm up, we illustrate the idea with a max register implementation,
and then show how it can be used to implement atomic snapshots.
Finally, we use atomic snapshots to obtain a general implementation of
\emph{simple types}~\cite{OvensW19,AspnesH1990wait},
which include counters, logical clocks and set objects.

\subsection{Max register}

To illustrate the idea, we first describe how to implement a
wait-free strongly-linearizable max register, using {\sf fetch\&add}.
Recall that a max register provides two operations
{\sf WriteMax}($v$) and {\sf ReadMax}() returning a value.
Its sequential specification is that a {\sf ReadMax} returns the
largest value previously written.

The key idea of the implementation is to pack into a single register $R$
an array containing the largest value written by each process.
A natural idea would be to give each process a consecutive set of bits from $R$,
i.e., $p_0$ gets bits $0,...,d-1$, $p_1$ gets bits $d,...,2d-1$, etc.
To modify its value, say to increment by 1,
$p_i$ would {\sf fetch\&add} an appropriate value.
The problem is that this bounds the values that can be written
by each process by $2^d$.

Instead, the idea is to \emph{interleave} the bits of the processes;
this representation was used in a recoverable implementation of
{\sf fetch\&add}~\cite{NahumABH2021}.
Specifically, $p_0$ stores its value in bits $0,n,2n,3n,...,$,
$p_1$ gets bits $1,n+1,2n+1,3n+1$,... etc.
This allows each process to store unbounded values.
Note that $R$ stores unbounded values, as well.

In the warm up example, we assume that each process stores
the largest value it has written so far \emph{in unary}.
The process holds its largest previous value
in a local variable {\it prevLocalMax}.
Here is {\sf MaxWrite}($R$,$K$) for $p_i$:
\begin{enumerate}
\item If $K$ is smaller than or equal to {\it prevLocalMax},
then {\sf fetch\&add}($R$,0) and return.
This {\sf fetch\&add} is not needed for correctness,
but it simplifies the linearization proof.

\item Set $p_i$'s bits ${\mathit prevLocalMax}+1,...,K$ to 1 by
{\sf fetch\&add} the appropriate number to $R$.
For example, if $K={\mathit prevLocalMax}+1$,
apply {\sf fetch\&add}($R$,$2^{Kn+i}$).

\item Set {\it prevLocalMax} to $K$.
\end{enumerate}

In {\sf ReadMax}, just read $R$ using {\sf fetch\&add}($R$,0),
reconstruct the individual maximums and return the largest one.

The linearization point of each operation is at its {\sf fetch\&add} operation,
immediately implying the implementation is strongly linearizable.
It is simple to obtain the following result:

\begin{theorem}
\label{thm-from-fetch-n-add-to-max-register}
There is a wait-free strongly-linearizable implementation
of a max register using fetch\&add.
\end{theorem}

\subsection{Atomic snapshots}

Next, we extend this idea to implement $n$-component (single-writer)
atomic snapshots; each component belongs to a single process,
and is initially $0$.
A \emph{view} is an $n$-component vector, each holding the value of a process.
An \emph{atomic snapshot} object~\cite{AADGMS93} provides two operations:
{\sf update} that modifies the process's component,
and {\sf scan} that returns a view.
The sequential specification of this object requires the view returned by
a {\sf scan} to include, in each component,
the value of the latest {\sf update} to this component.

In the implementation,
similar to the interleaving construction used for the max register,
a register $R$ holds a view so that the value of $p_i$ (the $i$th component)
is stored in bits $i,n+i,2n+i,3n+i,\ldots$.
In addition, each process has a local variable {\it prevVal},
which holds the current value stored in its component,
i.e., the value written in its previous update
(or the initial value, if no such update exists).

An {\sf update}($v$) operation by $p_i$ proceeds as follows.
\begin{enumerate}
\item If $v$ = {\it prevVal}, then {\sf fetch\&add}($R$,0) and return.

\item
Let $i_1,\ldots$ be the bits of $v$ that are 1,
and the bits of {\it prevVal} are 0 (these bits have to be set);
let $j_1,\ldots$ be the bits of $v$ that are 0,
and the bits of {\it prevVal} are 1 (these bits have to be unset).
Let
\[ {\mathit posAdj} = 2^{i_1 n + i}+2^{i_2 n + i}+\ldots \]
\[ {\mathit negAdj} = 2^{j_1 n + i}+2^{j_2 n + i}+\ldots \]

\item $p_i$ calls {\sf fetch\&add}($R$,$\mathit{posAdj}-\mathit{negAdj}$).

\item $p_i$ stores $v$ in {\it prevVal}.
\end{enumerate}

A {\sf scan} operation reads $R$ using {\sf fetch\&add}($R$,0),
reconstructs the view stored in it,
and returns it.

The linearization point of each operation is at its {\sf fetch\&add} operation,
immediately implying the implementation is strongly linearizable.
It is simple to obtain the following result:

\begin{theorem}
\label{thm-from-fetch-n-add-to-snapshots}
There is a wait-free strongly-linearizable implementation
of atomic snapshot using fetch\&add.
\end{theorem}

\subsection{Simple Types}

Aspnes and Herlihy~\cite{AspnesH1990wait} define objects
where any two operations $o_1$ and $o_2$ either \emph{commute}
(meaning the system configuration obtained after both operations
have been executed consecutively is independent of the order of
the two operations), or one of them \emph{overwrites} the other
(meaning that the system configuration obtained after the overwriting
operation has been performed is not affected by whether
or not the other operation is executed immediately before it).
Examples of such objects are a (monotonic and non-monotonic) counter
and a max register.
Aspnes and Herlihy show a wait-free linearizable implementation of
any such object, using one {\sf scan} and one {\sf update} operation.
Such objects are called \emph{simple types}
by Ovens and Woelfel~\cite{OvensW19},
who show that this implementation is strongly linearizable.
We next provide a simple proof of this fact,
which immediately yields a strongly-linearizable implementation
of simple types using {\sf fetch\&add},
by substituting our strongly-linearizable atomic snapshots.
See Algorithm~\ref{alg:simple types}.

\begin{algorithm}[tb]
\caption{\small Implementation of a simple type~\cite{AspnesH1990wait};
based on~\cite[Algorithm 5]{OvensW19}.}
\label{alg:simple types}
\begin{algorithmic}[1]\small
\Statex struct node :
\Statex  $\quad$ invocation description, $invocation\in O$
\Statex  $\quad$  response, $response \in R$
\Statex  $\quad$  pointers to nodes, $preceding[1 . . . n]$
\Statex shared atomic snapshot object $root = (null, . . . , null)$
\Statex
\Procedure{lingraph}{$G$}
\State    let $op_1, . . . , op_k$ be a topological sort of $G$
\State    $L$ = $G$
\For{$i \in \{1, . . . , k - 1\}$}
{\For{$j \in \{i + 1, . . . , k\}$}
{\If{$op_i$ dominates $op_j$ and adding ($op_j$ , $op_i$) to $L$ does not complete a cycle}
\State        add ($op_j$ , $op_i$) to $L$
\EndIf
\If{$op_j$ dominates $op_i$ and adding ($op_i$ , $op_j$) to $L$ does not complete a cycle}
\State        add ($op_i$ , $op_j$) to $L$
\EndIf}
\EndFor}\EndFor
\State return L
\EndProcedure
\Statex
\Procedure{execute$_p$}{invoke}: \Comment{Executing invoke on process $p$}
\State    view = $root$.scan() \label{line:scan}
\State    G = BFS/DFS traversal of the set of nodes starting from those in view
\State    S = topological sort of \textsc{lingraph}(G) \label{line:lin}
\State    initialize a new node e = $\{\bot,\bot,\bot\}$
\State    e.$invocation$ = invoke
\State    inv(op) = (invoke, id)
\State    rsp(op) = (resp, id) such that S $\circ$ inv(op) $\circ$ rsp(op) is valid
\State    e.$response$ = resp
\For{$i \in {1, . . . , n}$}
\State e.preceding[i] = view[i]
\EndFor
\State $root$.update$_p$(address of e) \label{line:write}
\State return e.$response$
\EndProcedure
\end{algorithmic}
\end{algorithm}

We concentrate on proving that the implementation is strongly linearizable,
as the other properties, e.g., correctness of return values, support the same arguments as in~\cite{AspnesH1990wait,OvensW19}.
We prove strong linearizability by defining a forward simulation from Algorithm~\ref{alg:simple types} to an atomic object (that implements the same type) which is defined as in~\cite[Page 6]{AttiyaE2019}. States of the atomic object are represented as sequential executions and transitions simply append operations to executions.
A \emph{forward simulation} $F$ is a binary relation between states of the two objects which is used to show inductively that any execution $\alpha$ of Algorithm~\ref{alg:simple types} can be mimicked by a sequential execution $S$ which is actually a linearization of the former. This relation is required to (1) relate any initial states of the two objects (base step of the induction), and (2) for every two related states $s_1$ (of Algorithm~\ref{alg:simple types}) and $s_2$ (of the atomic object), and every successor $s_1'$ of  $s_1$, there exists a state $s_2'$ of the atomic object which is reachable by 0 or more steps from $s_2$ and such that $s_1'$ and $s_2'$ are again related by $F$.

\begin{theorem}
\label{thm-from-snapshot-to-simple}
There is a wait-free strongly-linearizable implementation of any simple
type object from atomic snapshots.
\end{theorem}

\begin{proof}
Procedure \textsc{execute$_p$} in Algorithm~\ref{alg:simple types} maintains a graph where nodes represent invocations with their responses and edges represent a partial real-time order, i.e., the fact that an invocation started after another one finished (the predecessors of an operation in real-time are stored in the $preceding$ array). The snapshot object $root$ stores the last (maximal) operations in this real-time order. Every invocation builds a view of this graph by reading $root$ and traversing the $preceding$ fields until reaching $null$ values. Then, it defines a linearization of this graph (a sequential execution) starting from an arbitrary topological sort and according to a dominance relation between operations defined as follows: $o_1$ is dominated by $o_2$ if $o_2$ overwrites $o_1$ but not vice-versa, or $o_1$ and $o_2$ overwrite each other and $o_1$ is executed by a process with a smaller id. The returned response is the response possible by executing the invocation immediately after this linearization.

We say that a set of operations $O$ is \emph{transitively-dominated} by $o$ if for every $o'\in O$, there exist a set of operations $o_1=o'$,$\ldots$,$o_k$ in $O$, such that $o_i$ is dominated by $o_{i+1}$ for all $1\leq i<k$, and $o_k$ is dominated by $o$.

We define a forward simulation $F$ from Algorithm~\ref{alg:simple types} to the corresponding atomic object as follows. Let $C_1$ be a state of Algorithm~\ref{alg:simple types}, and let $G_1$ be the graph that it stores, i.e., containing all the nodes reachable from those in the snapshot $root$. Let $C_1^*$ be the state obtained from $C_1$ by making the following set of operations complete: operations that scanned the snapshot $root$ (line~\ref{line:scan}) before a write to $root$ (line~\ref{line:write}), i.e., their view of the graph misses at least one operation from $G_1$, and which are transitively-dominated by an operation already recorded in $G_1$. Let $G^*_1$ be the graph stored in $C_1^*$ (which may contain more nodes in comparison to $G_1$, namely, those that correspond to operations completed when going from $C_1$ to $C_1^*$).
The forward simulation $F$ associates $G^*_1$ with every linearization (sequential execution) defined as in line~\ref{line:lin} (a topological sort of $\textsc{lingraph}(G^*_1)$)

We now show that $F$ is indeed a forward simulation. The fact that it relates initial states is trivial. The interesting step is writing the snapshot $root$ (line~\ref{line:write}). The other steps are simulated by a stuttering step (0 steps) of the atomic object.
Let $C_1$ and $C_2$ be configurations of Algorithm~\ref{alg:simple types} before and after a write to $root$. Let $S_{C_1}$ be a sequential execution associated to $C_1$ as explained above.

If the current operation $o$ (executing the write) is already included in $S_{C_1}$ then this step is simulated by a stuttering step of the atomic object. This holds because the graph $G_2^*$ corresponding to $C_2$ is the same as $G_1^*$, and thus, $S_{C_1}$ is a sequential execution associated by $F$ to $C_2$ as well. Indeed, $o$ was included in $G_1^*$ because it was transitively dominated by an operation in $G_1$, and any operation transitively dominated by $o$ is also transitively dominated by an operation in $G_1$.

Otherwise, let $D_o$ be the set of operations which are (1) pending in $C_1$, scanned the snapshot $root$ (line~\ref{line:scan}), but did not write to $root$, and (2) transitively-dominated by $o$. This step of $o$ is simulated by the sequence $S$ of operations in $D_o\cup\{o\}$ ordered as in the definition of a linearization of a graph, i.e., a topological sort of \textsc{lingraph} applied to a graph containing the nodes corresponding to $D_o\cup\{o\}$. Applying these steps starting from the atomic object state $S_{C_1}$ leads to the state $S_{C_1}\cdot S$. It is easy to see that $S_{C_1}\cdot S$ is one of the sequential executions associated by $F$ to $C_2$.
\end{proof}

Combining this result with the implementation of
Theorem~\ref{thm-from-fetch-n-add-to-snapshots},
and relying on the fact that strong linearizability
can be composed~\cite[Theorem 10]{AttiyaE2019}.

\begin{theorem}
\label{cor-from-fetch-n-add-to-simple}
There is a wait-free strongly-linearizable implementation of any simple
type object from fetch\&add.
\end{theorem}

\section{Strongly-linearizable implementations of objects with consensus number 2}
\label{sec:algs cons 2}

In this section we exhibit examples of objects with consensus number 2 that can be
lock-free or wait-free implemented from test\&set or fecth\&add, while achieving
strong linearizability.

\subsection{Readable multi-shot test\&set}

We start by implementing readable test\&set.
In the implementation, the processes share a read/write register $state$, initialized to $0$,
and an $n$-process test\&set object $ts$. The $read()$ operation simply reads $state$ and returns the read value.
The $test\&set()$ operation first performs $ts.test\&set()$, then writes $1$ in $state$ and finally returns the value obtained from $ts$.

The key idea of the strong linearization proof of the implementation is that $state$ contains the state of the object at all times,
and when it changes from $0$ to $1$, all $test\&set()$ operations that have accessed $ts$ but have not written yet in $state$
are linearized at that write step, placing first the $test\&set()$ operation that obtains $0$ from $ts$.

\begin{restatable}{theorem}{ReadTS}
\label{thm-from-ts-to-read-ts}
There is a wait-free strongly-linearizable implementation
of readable test\&set using test\&set.
\end{restatable}

\begin{proof}
The implementation is clearly wait-free.
Also, it is clear that in every execution at most one test\&set operation returns $0$,
due to the specification of the object $ts$,
and if the execution has at least one correct process, only one test\&set operation returns $0$.

The operations are linearized as follows.
Read operations are linearized when they read $state$.
Test\&set operations are linearized in a slightly more complex way.
Let $e$ be the write that writes for the first time $1$ in $state$ (if any), and $op^*$ be the test\&set operation that obtains $0$ from $ts$ (if any).
Then $op^*$ is linearized at $e$, and right after $op^*$ all test\&set operations that
accessed $ts$ before $e$ are linearized; all other test\&set operations are linearized when they access $ts$.

Let $op$ be the test\&set operation $e$ belongs to.
We claim that the linearization points define a valid sequential execution.
First observe that if $op^* = op$, then $e$ is a step of $op^*$.
If not, $e$ is not an step of $op^*$, but it lies between the invocation and response of $op^*$:
if $e$ precedes the invocation of $op^*$, then $op^*$ could not get $0$ from $ts$,
and if the response of $op^*$ precedes $e$, then $e$ could not be first write that writes $1$ in $state$.
Similarly, $e$ lies between the invocation and response of any test\&set operation distinct from $op^*$ linearized at $e$:
if not, either $e$ is not the first write that writes $1$ in $state$, or the operation does not access $ts$ before $e$.
In any case, linearization points respect real-time order of operations.
The sequential execution is valid becase $state$ effectively reflects the state of the
implemented test\&set object: $op^*$ is placed at the point $state$ changes,
no other test\&set operation appears before $op^*$, and read operations return the value in $state$.

Finally, the linearization points define a prefix-closed linearization function because in any extension
of an execution it remains true that $e$ is the first write that writes $1$ in $state$ and $op^*$
is the test\&set operation that returns $0$.
\end{proof}

Next, we consider \emph{multi-shot} test\&set, which has an operation $reset()$ that resets the state of the objects to $0$.
We now provide a readable and multi-shot test\&set using as (atomic) base objects
readable test\&set and max register.
In the implementation, the processes share a max register object, $curr$, initialized to $1$,
and an infinite array, $TS$, of readable test\&set objects.
The $test\&set()$ and $read()$ operations simply return $TS[curr.readMax()].test\&set()$ and $TS[curr.readMax()].read()$, respectively.
The $reset()$ operation first performs $curr.readMax()$, storing the result in a local variable $c$,
then executes $TS[c].read()$, and if it obtains $1$, it does $curr.maxWrite(c+1)$.

The key idea in the strong linearizability proof is that the current
state of the object is that of $TS[v]$,
where $v$ is the current value in $curr$,
and a reset operation intends to reset the object only
if reads $1$ from $TS[v]$ (i.e., only if necessary).
Logically, the object is reset, hence transitions from
state $1$ to $0$, when a reset operation executes $curr.maxWrite(v+1)$ for the first time (several reset operations might
execute the same but only the first one has effect on the state of $curr$),
and when this even happens, every test\&set operation that has accessed $TS[v]$ is linearized.

\begin{restatable}{theorem}{MultiTS}
\label{thm-from-read-ts-to-ts}
There is a wait-free strongly-linearizable implementation
of readable and multi-shot test\&set using test\&set and max-register.
\end{restatable}

\begin{proof}
Consider the implementation described above.
The implementation is clearly wait-free. For the proof of strong linearizability,
observe that the state of $curr$ monotonically increments by one in any execution, as it evolves.
Suppose that at a given moment of time in an execution, $curr$  contains the value $v$.
The key idea is that the current state of the object is that of $TS[v]$ and a reset operations intends to reset the object only
if reads $1$ from $TS[v]$ (i.e., only if necessary).
Logically, the object is reset, hence transitions from
state $1$ to $0$, when a reset operation executes $curr.maxWrite(v+1)$ for the first time (several reset operations might
execute the same but only the first one has effect on the state of $curr$).
Let $e$ denote such event.
As long as $e$ does not occur, all operations are linearized at the moment
they access the objects in $TS$. When $e$ does occur, at that moment of time the following operations are linearized:
in any order, all tes\&set and read operations that have read $v$ from $curr$ (before $e$) but have not accessed $TS[v]$ yet
(note all such test\&set and read operations surely will obtain $1$ from $TS[v]$),
then the reset operation $e$ belongs to (which actually changes the state of $curr$),
followed by all reset operations that have read $v$ from $curr$ (before $e$) and but have
not read $TS[v]$ yet (all surely will obtain $1$ from $TS[v]$) or
have not executed  $curr.maxWrite(v+1)$ yet (none of them will affect the state of $curr$).

We first observe that every operation is linearized at an event that lies between its invocation and response,
hence the induced sequential execution respects operations' real-time order.
The interesting case is when a bunch of operations are linearized at an event like $e$ above.
Note that if $e$ precedes the invocation of any of these operations, then it could not read $v$ from $curr$,
and if the response of any of these operations precedes $e$, then it is not true that the
operations has not accessed $TS[v]$ or execute $curr.maxWrite(v+1)$ when $e$ happens.
Now we observe that the sequential execution is valid because between two consecutive
events $e$ and $e'$ when the object is reset,
there is one test\&set operation $op^*$ that returns $0$
($e'$ cannot exist if no operation obtains $0$ from the current object in $TS$),
all previous read operation to $op^*$ in the linearization return $0$ (any of them should access the current object
in $TS$ before $op^*$),
and all test\&set and read operations before the reset operations linearized at $e$ obtain $1$.
Finally, the linearization naturally define a prefix-closed linearization function as they were defined as
executions evolve in time.
\end{proof}

By Theorems~\ref{thm-from-fetch-n-add-to-max-register} and~\ref{thm-from-ts-to-read-ts},
there are wait-free strongly-linearizable implementations of
max registers and readable test\&set using fecth\&add and test\&set, respectively.
These implementations, combined with the implementation above (Theorem~\ref{thm-from-read-ts-to-ts}) give:

\begin{corollary}
There is a wait-free strongly-linearizable implementation
of readable and multi-shot test\&set using test\&set and fetch\&add.
\end{corollary}

If we consider the read/write strongly-linearizable max register implementation
in~\cite{HelmiHW2012, OvensW19},
which is only lock-free, then we obtain:

\begin{corollary}
There is a lock-free strongly-linearizable implementation
of readable and multi-shot test\&set using test\&set.
\end{corollary}

\subsection{Fetch\&Increment}

In the implementation, there is an infinite array $M$ of readable test\&set objects.
The $fetch\&increment()$ operation performs test\&set on each of the objects in $M$,
in index-ascending order, until it obtains $0$, and returns the index of the object.
Similarly, the $read()$ operation reads the objects in $M$,
in index-ascending order, until it obtains $0$ and returns the index of the object.

The key idea in the strong linearizability proof is that, at all times,
the state of the implemented object is the maximum index $i$
such that the state of the test\&set object $M[i]$ is 0.
Thus, any operation is linearized when it obtains $0$ (either using a read or test\&set)
from an object in~$M$.

\begin{restatable}{theorem}{FetchInc}
\label{thm-from-ts-to-fetch-inc}
There is a lock-free strongly-linearizable implementation
of readable fecth\&increment using test\&set.
\end{restatable}

\begin{proof}
It is not difficult to see that the implementation above is lock-free: the only reason a fetch\&increment or read operation
does not terminate is because infinitely many fetch\&increment terminate.
The key idea of linearizability is that, at all times, the state of the implemented object is the maximum index $i$
such that the state of the test\&set object $M[i]$ is 0.
Thus, any operation is linearized when it obtains $0$ (either using a read or test\&set) from an object in $M$.
The linearization points respect the real-time order of operations, and induce a valid sequential execution
as all processes access the objects in $M$ in the same order.
Furthermore, they induce a prefix-closed linearization function because the points are fixed.
The theorem follows due to this implementation and Theorem~\ref{thm-from-ts-to-read-ts}.
\end{proof}

\subsection{Sets}

We now consider sets. A set object provides two operations, $put(x)$ that adds item $x$ to the set and returns OK,
regardless if the item was already in the set or not, and $take()$ that returns EMPTY, otherwise
it returns and removes any item in the set.
For simplicity we assume that every item is the input of at most one put operation.

Consider the implementation in Algorithm~\ref{alg:from-ts-to-fi} of a set that uses as (atomic) base objects an infinite array of
test\&set objects, one readable fetch\&increment object and an infinite array of read/write objects.
The key idea in the strong linearizability proof is that, at any given moment in an execution,
the set contains any $x$ such that $Items[i] = x$, $1 \leq i \leq Max-1$ and $TS[i] = 0$,
where, by abuse of notation $Max$ and $TS[i]$ denote the state of the objects.
In words, the item has been placed somewhere in the active region of $Items$ and nobody has taken the item.
The operations are linearized as follows.
Put operations are linearized at the step they write in $Items$,
take operations that return an item are linearized at the event they obtain $0$ from the objects in $TS$,
and take operation that return empty are linearized at its last step that reads $Max$.

\begin{algorithm}[tb]
\caption{\small From test\&set to fetch\&increment. Algorithm for proces $p_i$.}
\label{alg:from-ts-to-fi}
\begin{algorithmic}[1]\small
\Statex Shared variables:
\Statex  \hspace{0.2cm} $Items =$ infinite array of read/write objects, each initialized to $\bot$
\Statex  \hspace{0.2cm} $TS =$ infinite array of test\&set objects, each initialized to $0$
\Statex  \hspace{0.2cm} $Max =$ readable fetch\&increment object, initialized to $1$
\Statex

\Procedure{Put}{$x$}
	\State $max = Max.fetch\&increment()$
	\State $Items[max].write(x)$
	\State \Return $OK$
\EndProcedure

\Statex

\Procedure{Take}{}
	\State $taken\_old = 0$
	\State $max\_old = 0$
    \While{$true$}
		\State $taken\_new = 0$
		\State $max\_new = Max.read()-1$
		\For{$c \in \{1, \hdots , max\_new\}$}
			\State $x = Items[c].read()$
			\If{$x \neq \bot$}
				\If{$TS[c].test\&set() == 0$}
					\Return $x$
				\EndIf
			\EndIf
		\EndFor
		\If{$taken\_new == taken\_old \, \wedge \, max\_new == max\_old$}
			\Return $EMPTY$
		\EndIf
		\State $taken\_old = taken\_new$
		\State $max\_old = max\_new$
\EndWhile
\EndProcedure
\end{algorithmic}
\end{algorithm}

\begin{restatable}{theorem}{Set}
\label{thm-from-ts-to-set}
There is a lock-free strongly-linearizable implementation
of a set using test\&set.
\end{restatable}

\begin{proof}
The implementation is lock-free because $Put(x)$ is clearly wait-free and the only reason
a $Take()$ of a correct process never terminates is because infinitely many put and take
operations are completed. We now argue that the operation is strongly linearizable.
For simplicity we assume that every item $x$ is input of a put operation at most once.\footnote{Otherwise the
implementation implements a multiset.}
First observe that, for every entry of $Items$, at most one put operation places its item in it, due to the specification of fetch\&increment.
Similarly, every item $x$ in $Items$ is returned by at most one take operation, due to the specification of test\&set.
The key idea in the strong linearizability proof is that, at any given moment in an execution,
the set contains any $x$ such that $Items[i] = x$, $1 \leq i \leq Max-1$ and $TS[i] = 0$,
where, by abuse of notation $Max$ and $TS[i]$ denote the state of the objects.
In words, the item has been placed somewhere in the active region of $Items$ and nobody has taken the item.
The operations are linearized as follows.
Put operations are linearized at the step they write in $Items$,
take operations that return an item are linearized at the event they obtain $0$ from the objects in $TS$,
and take operation that return empty are linearized at its last step that reads $Max$.
The linearization points define a sequential execution that respects the real-time order of operations
as every operation is linearized at one of its steps. The linearization is valid because, first, for a take operation
to return an item $x$, there must be a put operation before with input $x$, and second,
when a take operation returns empty, for every $1 \leq i \leq Max-1$, either $Items[i] = NULL$
or $Items[i] \neq NULL$ and $TS[i] = 1$.
Furthermore, the linearization points define a prefix-closed linearization function because
clearly linearization points do not change in any extension.

Finally, the theorem follows from the implementation just described and
the lock-free readable fetch\&increment implementation
in Theorem~\ref{thm-from-ts-to-fetch-inc}.
\end{proof}

\section{Impossibility of strongly-linearizable implementation from 2-consensus objects}
\label{sec:imposs cons 2}

We now show that lock-free strongly-linearizable implementations
of \emph{$k$-ordering} objects (defined below) imply solutions to
$k$-set agreement.
Recall that $k$-set agreement is a generalization of consensus
where processes are allowed to decide different proposed values
with the restriction that there are at most $k$ decided values.

More specifically, we show that $k$-set agreement among $n$ processes can be solved from
any lock-free strongly-linearizable $n$-process
implementation of a $k$-ordering object that uses readable base objects, namely, every base object provides a read operation.
The set agreement algorithm is inspired by the algorithm
of~\cite[Section~4, Figure~4]{AttiyaCH2018} that solves
consensus from any queue with \emph{universal helping}.
(Recall that, roughly speaking,
universal helping means eventually every pending operation is linearized.)
The algorithm we present here only assumes that read
operations return the current state of the base object.
In contrast, the algorithm in~\cite{AttiyaCH2018} assumes that read operations
return the \emph{complete history}, i.e.,
the full sequence of operations the base object has performed so far.

In what follows, for a sequential execution $\alpha$ of an object,
let $\alpha|i$ be the subsequence of $\alpha$ with the invocations and responses of $p_i$,
and let $invs(\alpha)$ and $resps(\alpha)$ be the subsequences of $\alpha$
with only invocations and responses, respectively.
For an object $O$, we let $Res(O)$ denote the set of responses of $O$,
and $Res(O)^+$ denote the set of all non-empty sequences over the elements of $Res(O)$.

Roughly speaking, an object is \emph{$k$-ordering} if for every process there is a pair of sequence
invocations, its \emph{proposal} sequence and its \emph{decision} sequence, and a \emph{decision}
function, such that the processes can solve $k$-set agreement by executing their proposal sequences
in a lock-free strongly-linearizable implementation of the object, where the decisions of the set agreement are made,
and then each process obtains its decision by locally simulating its decision sequence, and with the help of the decision function.

\begin{definition}
\label{def-agreement-objs}
A sequential object $O$ is \emph{$k$-ordering}, $1 \leq k \leq n-1$, if
there are~$2n$ non-empty sequences of invocations of $O$, $prop_0, dec_0, \hdots, prop_{n-1}, dec_{n-1}$
and a function $d: \{0, \hdots, n-1\} \times Res(O)^+ \rightarrow  \{0, \hdots, n-1\}$
such that, for every sequential execution $\alpha$ of $O$ with
\begin{itemize}
\item $invs(\alpha) \subseteq \bigcup^{n-1}_{x=0} invs(prop_x)$ and
\item $invs(\alpha|j) = prop_j$ for some $p_j$,
\end{itemize}
there is a set $S_\alpha$ with $k$ process indexes such that, for every process $p_i$,
for every sequential  execution $\alpha \cdot \alpha' \cdot  \beta_i$ of $O$ with
\begin{itemize}
\item $invs(\alpha \cdot \alpha') \subseteq \bigcup^{n-1}_{x=0} invs(prop_x)$,
\item $invs((\alpha \cdot \alpha')|i) = prop_i$ and
\item $invs(\beta_i) = dec_i$,
\end{itemize}
there is a process $p_\ell$ such that
\begin{itemize}
\item $invs((\alpha \cdot \alpha')|\ell) = prop_\ell$ and
\item $d(i, resps((\alpha \cdot \alpha')|i) \cdot resps(\beta_i)) = \ell \in S_\alpha$.
\end{itemize}
\end{definition}

\begin{lemma}
\label{lemma-consensus}
Let $A$ be a lock-free strongly-linearizable $n$-process implementation of a $k$-ordering object $O$,
on top of a system with readable base objects
where a read operation returns the current state of the object.
Then, $k$-set agreement among $n$ processes can be solved in the same system using a single instance of~$A$.
\end{lemma}

\begin{proof}
We present an algorithm $B$ that solves $k$-set agreement among $n$ processes using $A$.
Let $prop_0, \hdots, prop_{n-1}, dec_0, \hdots, dec_{n-1}$ be the sequences of invocations and $d$ be the function
guaranteed to exist due to the assumption that $O$ is a $k$-ordering object.
The idea of the algorithm is simple.
Every process $p_i$ performs on $A$ all its invocations in $prop_i$.
Since $A$ is strongly linearizable, at the end of its operations,
the order of the operations is fixed. Then $p_i$ takes a snapshot
of $A$'s base objects, which is possible since base objects are readable,
and finally it locally executes all its invocations in $dec_i$ and, with the help of function $d$,
obtains one of the winners of the $k$-set agreement.

More concretely, in algorithm $B$, the processes use two shared arrays $M$ and $T$ of length~$n$ each,
with every entry initialized to $\bot$,
and the set, denoted $R$, with all base objects accessed in all executions of $A$ where every process $p_i$
executes some or all invocations in $prop_i$, in the order specified in the sequence.
Observe that $R$ is finite as there are finitely many such executions, since $A$ is lock-free and each $prop_i$ is finite.
Also, the processes use a function $collect(S)$, where $S$ is either a set or an array,
that reads one by one, in any arbitrary order,
the base objects in $S$.
In $B$, process $p_i$ with input $x$ does:
\begin{enumerate}

\item Set a local variable $t$ to $0$.

\item Execute $M[i].write(x)$.

\item Execute one by one, until completion and  in order, every invocation in $prop_i$ as follows:

\begin{enumerate}

\item $t \leftarrow t+1$.

\item $T[i].write(t)$.

\item Execute the next step of $p_i$ in the current invocation of $prop_i$, as dictated by the algorithm.

\end{enumerate}

\item do

\begin{enumerate}

\item $t_1 \leftarrow collect(T)$

\item $r \leftarrow collect(R)$

\item $t_2 \leftarrow collect(T)$

\end{enumerate}

\item while $\exists j, t_1[j] \neq t_2[j]$

\item Starting from the states of base objects in~$r$, simulate in order and until completion
all invocations in $dec_i$.

\item Return $M[d(i, resps)].read()$, where $resps$ is the sequence of response from $A$ obtained in Steps 3 and 6.
\end{enumerate}

We prove that $B$ satisfies termination, validity and $k$-agreement.
Below, for any execution $F$ of $B$,
let $F|A$ denote the execution of $A$ in $F$, i.e., the subsequence of $F$ with the steps of $A$.

First observe that any correct process terminates Step~3 of $B$ because $A$ is lock-free and
every process executes at finite number of operations.
Thus, we have that, in every execution, every process either completes Step~3 or crashes.
Note that this implies that
every correct process terminates the loop in Steps~4 and~5.

For the rest of the proof, fix en execution $E$ of $B$ with at least one correct process,
and let $E'$ be the shortest prefix of $E$ where a process terminates its operations in Step~3.
For a process $p_i$ that completes its loop in Steps~4 and~5,
let $E_i$ the shortest prefix of $E$ where $p_i$ completes the loop.
Hence~$E'$ is a prefix of~$E_i$.
We show that the states the base objects in $r$
that $p_i$ collects at the end of $E_i$ are a snapshot of $R$, in some extension
of $E'|A$, possibly distinct to $E_i|A$.
As explained later, this property implies termination, validity and $k$-agreement.

\begin{restatable}{claim}{Snapshots}
\label{claim-snapshots}
For every process $p_i$  that completes its loop in Steps~4 and~5 in~$E$,
there in an execution $E^*_i$ of $B$ such that
(1) $E^*_i|A$ is an extension of~$E'|A$,
(2) $E^*_i|A$ has only invocations that appear in the sequences $prop_0, \dots, prop_{n-1}$,
(3) $p_i$~has no pending operation in $E^*_i|A$, and
(4) the states in $r$ at the end of $E_i$ are a snapshot of $R$,
namely, they are the states
of the base objects in $R$ in the configuration at the end of~$E_i$.
\end{restatable}

\begin{proof}[Proof of the claim.]
For simplicity, and without loos of generality, let us assume that $p_i$ is the only process that
executes steps of the loop in Steps~4 and~5 in~$E$ (or alternatively, remove from $E$ all steps corresponding
to the loop of any process distinct from $p_i$).
Let $F$ the shortest prefix of $E$ such that $p_i$ completes the loop.

Intuitively, we will focus on the last iteration of the loop, and modify $F$ so that
the states in $r$ are a snapshot.
We start with the following observation. Let $e_1$ and $e_2$ be the last and first reads of the last two collects
of $p_i$ of $T$ in $F$, respectively (which are steps of $p_i$'s last iteration of the loop).
Observe that every process distinct from $p_i$ takes \emph{at most} one step (of $A$) between
$e_1$ and $e_2$ in $F$. The reason is that, in $B$, every process increments its entry in~$T$ before taking a step of $A$.
Let $H$ be the subsequence of $F$ with the steps between $e_1$ and $e_2$.
For each base object $X \in R$, let $r_X$ be the step of $p_i$ in $H$ that reads $X$,
and let $\beta_X$ be the subsequence of $H$ with the steps accessing $X$ after $r_X$.
Consider the sequence $E^*_i$ obtained from $F$ by
(1)~removing every $\beta_X$,
(2)~removing every step of a process distinct from $p_i$ that appears after $e_2$, and
(3)~``moving forward" all $r_X$, $X \in R$, just before $e_2$.

We claim that $E^*_i$ is an execution of $B$. For each process distinct from $p_i$,
it last steps of $A$ is removed, if any, which is possible due to asynchrony.
Due to asynchrony too, the reads $r_X$ of $p_i$ can be ``postponed" just before $e_2$.
Also note that the states of the objects in $R$ in the configuration at the end of $E^*_i$ are precisely those in $r$,
and $E'|A$ is a prefix of $E^*_i|A$, as $F$ and $E^*_i$ are equal up to $e_1$.
Finally, it is easy to see that $E^*_i|A$  has only invocations that appear in the sequences $prop_0, \dots, prop_{n-1}$
and $p_i$ has no pending operation in it. The claim follows.
 \end{proof}

We now prove termination, validity and $k$-agreement,
showing that $B$ solves $k$-set agreement.

Termination. Let $p_i$ be a correct process in $E$.
As already argued, $p_i$ complete its loop in Steps~4 and~5.
Consider an execution $E^*_i|A$ of $A$ as stated in Claim~\ref{claim-snapshots}.
As $r$ contains the states of the objects in $R$ at the end of $E^*_i|A$
and $p_i$ has no pending operation,
indeed $p_i$ is able to locally simulate a solo execution that extends $E^*_i|A$, in Step~6.
The local simulation terminates because $A$ is lock-free.
Thus, $p_i$ eventually makes a decision in Step~7.

Validity. For the rest of the proof, let $f$ be a prefix-closed linearization function of $A$
and let $\alpha = f(E'|A)$.
As  $E'|A$ is a prefix of $E^*_i|A$, we have that $\alpha$ is a prefix of $f(E^*_i|A)$,
hence $f(E^*_i|A)$ can be written $\alpha \cdot \alpha'$, for some $\alpha'$.
By definition of $E'|A$, we have $invs(\alpha|j) = prop_j$, for some process $p_j$.
Also observe that $invs((\alpha \cdot \alpha')|i) = prop_i$.
Let $F_i$ be the (solo) extension of $E^*_i|A$ that $p_i$ locally simulates in Step~6.
Note that $f(F_i)$ can be written $\alpha \cdot \alpha' \cdot \beta_i$,
for some $\beta_i$ with $invs(\beta_i) = dec_i$.
Let $\ell = d(i, resps(\alpha \cdot \alpha')|i) \cdot resps(\beta_i))$.
The properties of $d$ guarantee that $invs((\alpha. \cdot \alpha')|\ell) = prop_\ell$, which implies
that $p_\ell$ completes the loop in Step~3 in $E^*_i|A$, and hence the proposal of $p_\ell$
is in $M$ at the end of $E^*_i|A$. Hence, $p_i$ decides the input of $p_\ell$.

$k$-Agreement. By definition of $d$, there is a set $S_\alpha$ with at most
$k$ indexes of processes such that
$d(i, resps(\alpha \cdot \alpha')|i) \cdot resps(\beta_i)) \in S_\alpha$.
The set $S_\alpha$ depends on $\alpha = f(E'|A)$ and
$E'|A$ is prefix of all extensions the processes locally simulate.
Thus, there are at most $k$ distinct decisions.
\end{proof}

We thus have the following:

\begin{corollary}
Among systems with readable base objects, $n$-process lock-free strongly-linearizable implementations of
$k$-ordering objects can exists only in those in which $k$-set agreement among $n$ processes can be solved.
\end{corollary}

The following are examples of ordering objects:

\begin{itemize}

\item Queues are 1-agreement objects.
For each process $p_i$, $prop_i = enq(i)$, $dec_i = deq()$ and $d(i, OK \cdot \ell) = \ell$.
Consider any sequential execution $\alpha$ of the queue
with invocations in  $prop_0, \hdots, prop_{n-1}$ and $invs(\alpha|j) = prop_j$, for some process $p_j$.
Let $S_\alpha = \{\ell\}$ be the singleton set with the input (process index) of the first enqueue in $\alpha$.
Note that $invs(\alpha|\ell) = prop_\ell$.
For a process $p_i$, consider any sequential execution $\alpha \cdot \alpha' \cdot \beta_i$ of the
queue such that all invocation of $\alpha \cdot \alpha'$ appear in $prop_0, \hdots, prop_{n-1}$,
$invs((\alpha \cdot \alpha')|i) = prop_i$ and $invs(\beta_i) = dec_i$.
We have that $invs(\alpha \cdot \alpha'|\ell) = prop_\ell$.
Note that $\beta_i$ has a single dequeue that returns the input of the first enqueue in $\alpha \cdot \alpha'$,
which is $\ell \in S_\alpha$. Moreover, $resps((\alpha \cdot \alpha')|i) \cdot resps(\beta_i) = OK \cdot \ell$,
and hence $d(i, resps((\alpha \cdot \alpha')|i) \cdot resps(\beta_i)) = \ell \in S_\alpha$.\\

\item Stacks are 1-agreement objects.
For each process $p_i$, $prop_i = push(i)$, $dec_i = pop() \cdot \hdots \cdot pop()$ ($n+1$ times pop)
and $d(i, OK \cdot j_0 \cdot \hdots  \cdot  j_n) = j_x$,
where $j_x$ is the non-$\epsilon$ element of the sequence with largest subindex.
Consider any sequential execution $\alpha$ of the stack
with invocations in  $prop_0, \hdots, prop_{n-1}$ and $invs(\alpha|j) = prop_j$, for some process $p_j$.
Let $S_\alpha = \{\ell\}$ be the singleton set with the input (process index) of the first push in $\alpha$.
Note that $invs(\alpha|\ell) = prop_\ell$.
For a process $p_i$, consider any sequential execution $\alpha \cdot \alpha' \cdot \beta_i$ of the
stack such that all invocation of $\alpha \cdot \alpha'$ appear in $prop_0, \hdots, prop_{n-1}$,
$invs((\alpha \cdot \alpha')|i) = prop_i$ and $invs(\beta_i) = dec_i$.
We have that $invs(\alpha \cdot \alpha'|\ell) = prop_\ell$.
Observe that $\alpha \cdot \alpha'$ has at most $n$ push operations, and hence
among the $n+1$ pop operations in $\beta_i$, one of them returns $\epsilon$,
and all subsequent operations return $\epsilon$ too.
Moreover, notice that the last non-$\epsilon$ return value is $\ell \in S_\alpha$,
and thus $resps((\alpha \cdot \alpha')|i) \cdot resps(\beta_i) =
OK \cdot j_0 (\neq \epsilon) \cdot j_1 (\neq \epsilon) \cdot \hdots \cdot \ell \cdot \epsilon \cdot \hdots \cdot \epsilon$,
and hence $d(i, resps((\alpha \cdot \alpha')|i) \cdot resps(\beta_i)) = \ell \in S_\alpha$.\\

\item Queues and stacks with multiplicity~\cite{CastanedaRR23} are 1-agreement objects.
These are relaxations of queues and stacks where concurrent dequeue or pop operations, respectively,
are allowed to return the same item in the structure, and these operations are linearized
one after consecutively.\footnote{Queues and stacks with multiplicity are defined in~\cite{CastanedaRR23}
through set linearizability, a generalization of linearizability where concurrent operations might be linearized
together at the same linearization point. In these relaxations concurrent dequeues/pops returning the same
item are linearized together. Here we opt for the alternative definition using linearizability, with the
constraint that the relaxation can happen only in case of concurrent operations, a property
that is inherent to the set linearizable version.}
Since the relaxation can happen only if there are concurrent enqueue/pop operations,
the respective sequences and functions for queues and stacks above
work for queues and stacks with multiplicity too.

\item $m$-Stuttering queues~\cite{HenzingerKPSS13} are 1-agreement objects.
This is a relaxation of queues where operations might have no effect, i.e.,
they do not change the state of the object. For example,
an enqueue operation might return $OK$, although its item is not enqueued.
Similarly, a dequeue returns the oldest item in the queue, although it is not removed.
We consider stuttering where  enqueues and dequeues are treated independently,
and the relaxation might happen up to $m \geq 1$ times in each case, and once an operation
does have effect, the relaxation might occur again.\footnote{Formally,
the state of the object has a counter per operation type, and if the corresponding counter is less than $m$,
the objects non-deterministically decides whether  the operations has effect or not, and if it takes effect, the counter is set to zero.}
Thus, at least one out of $m+1$ consecutive operations of the same type is guaranteed to have effect.

For each process $p_i$, $prop_i = enq(i) \cdot \hdots \cdot enq(i)$ ($m+1$ times enqueue),
$dec_i = deq()$ and $d(i, OK \cdot \hdots \cdot OK \cdot \ell) = \ell$ ($m+1$ times $OK$).
Consider any sequential execution $\alpha$ of the stuttering queue
with invocations in  $prop_0, \hdots, prop_{n-1}$ and $invs(\alpha|j) = prop_j$, for some process $p_j$.
Note that at least one enqueue in $\alpha$ indeed enqueued its item, thus the state of the queue
at the end of $\alpha$ is not empty.
Let $S_\alpha = \{\ell\}$ be the singleton set with the input (process index) of the first enqueue in $\alpha$
that \emph{has effect}.
Note that $invs(\alpha|\ell) \neq \epsilon$ (an possibly $\neq prop_\ell$).
For a process $p_i$, consider any sequential execution $\alpha \cdot \alpha' \cdot \beta_i$ of the
stuttering queue such that all invocation of $\alpha \cdot \alpha'$ appear in $prop_0, \hdots, prop_{n-1}$,
$invs((\alpha \cdot \alpha')|i) = prop_i$ and $invs(\beta_i) = dec_i$.
We have that $invs(\alpha \cdot \alpha'|\ell) \neq \epsilon$.
Note that $\beta_i$ has a single dequeue that returns the input of the first enqueue in $\alpha \cdot \alpha'$ that has effect,
which is $\ell \in S_\alpha$. Moreover, $resps((\alpha \cdot \alpha')|i) \cdot resps(\beta_i) = OK \cdot \hdots \cdot OK \cdot \ell$ ($m+1$ times $OK$),
and hence $d(i, resps((\alpha \cdot \alpha')|i) \cdot resps(\beta_i)) = \ell \in S_\alpha$.\\

\item $m$-Stuttering stacks~\cite{HenzingerKPSS13} are 1-agreement objects.
This relaxation of stacks are defined similarly. Operation types are treated separately, and
the relaxation might happen up to $m \geq 1$ times in each case.

For each process $p_i$, $prop_i = push(i) \cdot \hdots \cdot push(i)$ ($m+1$ times push),
$dec_i = pop() \cdot \hdots \cdot pop()$ ($n (m+1)+1$ times pop)
and $d(i, OK \cdot \hdots \cdot OK \cdot  j_1 \cdot \hdots  \cdot  j_{n(m+1)+1}) = j_x$ ($m+1$ times $OK$),
where $j_x$ is the non-$\epsilon$ element of the sequence with largest subindex.
Consider any sequential execution $\alpha$ of the stuttering stack
with invocations in  $prop_0, \hdots, prop_{n-1}$ and $invs(\alpha|j) = prop_j$, for some process $p_j$.
Note that at least one push in $\alpha$ indeed pushed its item, thus the state of the stack
at the end of $\alpha$ is not empty.
Let $S_\alpha = \{\ell\}$ be the singleton set with the input (process index) of the first push in $\alpha$ that \emph{has effect}.
Notice $invs(\alpha|\ell) \neq \epsilon$.
For a process $p_i$, consider any sequential execution $\alpha \cdot \alpha' \cdot \beta_i$ of the
stuttering stack such that all invocation of $\alpha \cdot \alpha'$ appear in $prop_0, \hdots, prop_{n-1}$,
$invs((\alpha \cdot \alpha')|i) = prop_i$ and $invs(\beta_i) = dec_i$.
Note that $invs(\alpha \cdot \alpha'|\ell) \neq \epsilon$.
Observe that $\alpha \cdot \alpha'$ has at most $n(m+1)$ push operations that have effect, and hence
among the $n(m+1)+1$ pop operations in $\beta_i$, one of them returns $\epsilon$,
and all subsequent operations return $\epsilon$ too.
Moreover, notice that the last non-$\epsilon$ return value is $\ell \in S_\alpha$,
and thus $resps((\alpha \cdot \alpha')|i) \cdot resps(\beta_i) =
OK \cdot \hdots \cdot OK \cdot j_0 (\neq \epsilon) \cdot j_1 (\neq \epsilon) \cdot \hdots \cdot \ell \cdot \epsilon \cdot \hdots \cdot \epsilon$,
and hence $d(i, resps((\alpha \cdot \alpha')|i) \cdot resps(\beta_i)) = \ell \in S_\alpha$.\\

\item $k$-Out-of-Order queues~\cite{HenzingerKPSS13} are $k$-ordering objects, $1 \leq k \leq n-1$.
This is a relaxation of queues where a dequeue operation returns one of the $k$ oldest items in the queue
(hence $1$-out-of-order queues are just regular queues).

For each process $p_i$, $prop_i = enq(i)$, $dec_i = deq()$ and $d(i, OK \cdot \ell) = \ell$.
Consider any sequential execution $\alpha$ of the $k$-out-of-order queue
with invocations in  $prop_0, \hdots, prop_{n-1}$ and $invs(\alpha|j) = prop_j$, for some process $p_j$.
Let $S_\alpha$ be the set with the input (process index) of the first $k$ enqueues in $\alpha$.
Note that $1 \leq |S_\alpha| \leq k$. Also note that $invs(\alpha|\ell) = prop_\ell$, for each $\ell \in S_\alpha$.
For a process $p_i$, consider any sequential execution $\alpha \cdot \alpha' \cdot \beta_i$ of the
$k$-out-of-order queue such that all invocation of $\alpha \cdot \alpha'$ appear in $prop_0, \hdots, prop_{n-1}$,
$invs((\alpha \cdot \alpha')|i) = prop_i$ and $invs(\beta_i) = dec_i$.
We have that $invs(\alpha \cdot \alpha'|\ell) = prop_\ell$, for each $\ell \in S_\alpha$.
Notice $\beta_i$ has a single dequeue that returns the input of one of the first $k$ enqueues in $\alpha \cdot \alpha'$,
which belongs to $S_\alpha$. Moreover, $resps((\alpha \cdot \alpha')|i) \cdot resps(\beta_i) = OK \cdot \ell$,
and hence $d(i, resps((\alpha \cdot \alpha')|i) \cdot resps(\beta_i)) = \ell \in S_\alpha$.\\

\end{itemize}

We consider now systems where each base objects belong to the class of \emph{interfering} objects~\cite[Section~3.2]{H91}
(see also~\cite[Section~5.7]{HS08book}); this class of objects provide a combination of
read, write, test\&set, swap and fetch\&add operations, among others.
Theorem~4 in~\cite[Section~3.2]{H91} shows that it is impossible to solve three-process consensus
using interfering objects. This theorem implies:

\begin{corollary}
\label{coro-no-cons-read-2-RMW}
It is impossible to solve consensus among three or more processes using readable base objects of the type
test\&set, swap and fetch\&add.
\end{corollary}

The following lemma observes that any implementation that is strongly linearizable, remains so
in a system where the same base objects are readable.

\begin{restatable}{lemma}{StrongReadable}
\label{lemma-strong-readable}
Let $A$ be a strong linearizable algorithm that implements a type $T$ in a system with base objects in set $S$.
Then, $A$ is a strong linearizable implementation of $T$ in the system where every object in $S$ is readable,
and satisfying the same progress properties.
\end{restatable}

\begin{proof}
First note that $A$ is a linearizable implementation of $T$ in readable $S$ because,
by assumption, it is a linearizable implementation of $T$ in $S$,
and the only difference between the models is that the base objects in readable $S$ provide
read operations. Thus, any execution of $A$ in one of the models is an execution of $A$ in the other model.

Let $f$ be any prefix-closed linearization function of $A$ in $S$,
whose existence is guaranteed due to the fact that $A$ is strongly linearizable.
Consider any execution of $A$ in readable $S$. As explained,
$E$ is an execution of $A$ in $S$. We have that $f(E)$ is a linearization of $E$.
As $f$ is prefix-closed, we thus have that $A$ is strongly linearizable.
\end{proof}

We can now prove the impossibility of strongly-linearizable implementations of $1$-agreement objects.

\begin{theorem}
For $n \geq 3$, there is no lock-free strongly-linearizable implementation of a $1$-agreement object
using (non-readable) test\&set, swap and fetch\&add.
\end{theorem}

\begin{proof}
By contradiction, suppose that there is such an algorithm $A$.
We have hat $A$ is a lock-free strongly-linearizable queue in a system where test\&set, swap and fetch\&add
are readable, by Lemma~\ref{lemma-strong-readable},
and hence one can solve consensus among three or more processes using base objects of type test\&set, swap and fetch\&add,
by Lemma~\ref{lemma-consensus}.
But this contradicts Corollary~\ref{coro-no-cons-read-2-RMW}, as consensus among three or more processes
is impossible in this system.
\end{proof}

We consider now systems where base objects are of type $2$-process {\sf test\&set}.
We also consider systems where the object is readable.
It has been shown~\cite{HerlihyR94} that there is no algorithm
that solves $k$-set agreement using  $2$-process consensus objects whenever $n > 2k$.
Since $2$-process test\&set and $2$-process consensus are equivalent~\cite{H91},
this impossibility result extends to $2$-process test\&set.

\begin{corollary}
\label{coro-impossibility-k-set-from-2-ts}
If $n > 2k$, it is impossible to solve $n$-process $k$-set agreement using $2$-process test\&set.
\end{corollary}

We can prove the impossibility of lock-free strongly-linearizable $k$-ordering objects
from $2$-process test\&set, whenever $n > 2k$.

\begin{theorem}
If $n > 2k$, there is no lock-free strongly-linearizable implementation
of a $k$-ordering object using $2$-process test\&set.
\end{theorem}

\begin{proof}
By contradiction, suppose there is such an implementation $A$.
By Lemma~\ref{lemma-strong-readable}, $A$ is a lock-free strongly-linearizable implementation of
the same object using readable using $2$-process test\&set.
By Lemma~\ref{lemma-consensus}, we know that there is an algorithm $B$ that solves $k$-set agreement
using $A$. Hence $k$-set agreement is solvable using readable $2$-process test\&set.
If we replace the base objects in $A$
with the wait-free strong linearizable implementations of Theorem~\ref{thm-from-ts-to-read-ts},
$B$ remains correct because, first, $B$ heavily relies on the strong linearizability
of $A$, which is preserved as the implementations are strongly linearizable, and second,
termination of $B$ is preserved too because the implementations are wait-free.
Thus, $k$-set agreement is solvable using $2$-process test\&set..
But this contradicts Corollary~\ref{coro-impossibility-k-set-from-2-ts}.
\end{proof}

\section{Discussion}

We have studied whether primitives with consensus number 2
allow to obtain strongly-linearizable implementations of objects.
Naturally, we concentrate on objects with consensus number 1 or 2,
and show that for many of them,
there are wait-free strongly-linearizable implementations
from {\sf test\&set} and {\sf fetch\&add}.
We also prove that even
when {\sf fetch\&add}, {\sf swap} and {\sf test\&set} primitives are used,
some objects with consensus number 2,
like queues and stacks, and even their relaxed variants,
do not have lock-free strongly-linearizable implementations.

Our results indicate several intriguing research directions.
Immediate questions are to complete the picture for other objects
with consensus number 2, e.g., to find
a wait-free strongly-linearizable implementation of {\sf fetch\&inc}
from {\sf test\&set},
and wait-free, or even lock-free, strongly-linearizable implementations
of {\sf fetch\&add} or {\sf swap} from {\sf test\&set}.
Note that analogous linearizable implementations
exist~\cite{AGM07,AfekMW2011,Afek1999W,AfekWW1993,Li01},
but they are intricate and not strongly linearizable.
Moreover, we would like to find a characterization
of the objects (for each consensus number) that have
strongly-linearizable implementation from primitives with consensus number 2.
This characterization does not contain the class \emph{Common2},
since a stack is in Common2~\cite{AGM07},
but we show it does not have a wait-free strongly-linearizable
implementation from primitives with consensus number 2.

Our implementations using {\sf fetch\&add} store extremely large
values in a single variable.
It is interesting for find an implementation that uses smaller
variable, e.g., with only $O(\log n)$ bits.
One way is to find a strongly-linearizable implementation of \emph{wide}
{\sf fetch\&add} objects from narrow {\sf fetch\&add} objects,
or to show that such an implementation does not exist.

Our impossibility proofs go through a reduction, showing that
strongly-linearizable (and lock-free) implementations allow to
solve agreement problems.
Can such reductions be used to prove the impossibility
of strongly-linearizable implementations of objects with consensus number 1
from primitives with consensus number 1,
like the results of~\cite{DenysyukW2015,HelmiHW2012,ChanHHT2021}?
Such proofs by reduction would provide more insight into the reasons for
the difficulty of achieving strong linearizability.

\begin{acks}
Hagit Attiya is partially supported by the Israel Science Foundation (grant number 22/1425).
Armando Castañeda is supported by the DGAPA PAPIIT project IN108723.
Constantin Enea is partially supported by the project SCEPROOF founded by the French ANR Agency and the NSF Agency from USA.
\end{acks}

\bibliographystyle{abbrv}
\bibliography{biblio}

\end{document}